%% file: acl_latex.tex
\documentclass[11pt]{article}

\usepackage[final]{acl}

\usepackage{times}
\usepackage{latexsym}
\usepackage{amssymb} 
\usepackage[T1]{fontenc}
\usepackage[utf8]{inputenc}

\usepackage{microtype}

\usepackage{inconsolata}

\usepackage{graphicx}
\usepackage{algorithm}
\usepackage{algpseudocode}
\usepackage{textcomp}
\usepackage{xcolor}
\usepackage{amsthm}
\usepackage{amsmath}
\usepackage[caption=false,font=footnotesize]{subfig}
\usepackage{booktabs}
\newtheorem{proposition}{Proposition}
\title{RouterWise: Deployment-Aware Prompt Routing for Multi-Model LLMs

}

\author{
Hossein Hosseini Kasnavieh$^{1,2}$ \quad
Gholamreza Haffari$^{3}$ \quad
Chris Leckie$^{2}$ \quad
Adel N. Toosi$^{1,2}$ \\
$^{1}$DisNet Lab, School of Computing and Information Systems, The University of Melbourne, Australia \\
$^{2}$School of Computing and Information Systems, The University of Melbourne, Australia \\
$^{3}$Department of Data Science \& AI, Monash University, Australia \\
\texttt{hossein.hosseini@student.unimelb.edu.au} \\
\texttt{\{caleckie, adel.toosi\}@unimelb.edu.au} \\
\texttt{gholamreza.haffari@monash.edu}
}

\begin{document}
\maketitle

\input{abstract}
\input{introduction}
\input{preliminaries}
\input{problem_formulation}
\input{methodology}
\input{Evaluation}
\input{conclusion}

\input{limitation}
\bibliography{custom}

\input{appendix}

\end{document}

%% file: abstract.tex
\begin{abstract}
Multi-model LLM routing aims to direct each prompt to a model
whose capability matches its difficulty, balancing quality against
cost and latency. Existing routers are trained from prompts alone
and treat each candidate model as a black box with fixed
per-request latency. In real deployments, however, latency is not
a fixed property of a model: it depends on how the model is
deployed, its allocated compute and memory, and on the load induced by the routing policy
itself. The same model can be the right choice for a prompt under
one deployment and the wrong choice under another, and a router
calibrated on fixed latencies can produce assignments that are
slow, infeasible, or simply suboptimal once deployed.
We propose \textsc{RouterWise}, a deployment-aware router that
conditions routing decisions on the model's deployment
configuration. Given a pool of models, a GPU cluster, and a
latency SLO, \textsc{RouterWise} jointly selects a deployment
configuration and a routing policy that maximize output quality
under the SLO, using a dual-price formulation that reduces
score-maximizing routing to per-prompt decisions and
setup-specific latency models learned from system profiling.
Experiments on three open-source LLMs show that achievable routing
quality varies by up to 87\% across deployment configurations on
the same GPU cluster under the same SLO, demonstrating that the
optimal model for a prompt depends not only on the prompt itself,
but on the deployment environment.
\end{abstract}

%% file: introduction.tex
\section{Introduction}
Large Language Models (LLMs) have demonstrated remarkable
capabilities across reasoning, coding, and multi-step problem
solving~\cite{hetis, gllm}, and are now available in a wide range
of sizes from both proprietary providers~\cite{OpenAI2025,
anthropic2024} and the open-source community~\cite{Llama4_2025,
yang2025qwen3}. These models differ substantially in quality,
cost, and latency, motivating a central question in LLM serving:
how should providers trade output quality against the cost and
latency of producing it, given that user prompts themselves vary
widely in complexity? Prompt routing addresses this question by
dynamically selecting an appropriate model for each input,
sending simpler prompts to smaller, faster models and escalating
harder queries to larger ones, navigating the quality--cost
frontier on a per-prompt basis~\cite{best-route, zooter}.

\begin{figure}[t]
    \centering
    \subfloat[2 GPUs]{
        \includegraphics[width=0.46\columnwidth]{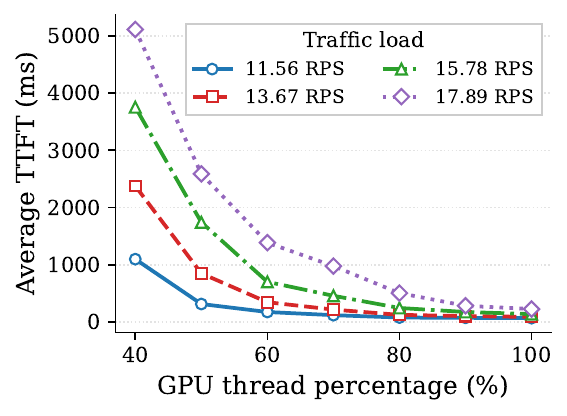}
        \label{fig:subfig1}
    }
    \hfill
    \subfloat[4 GPUs]{
        \includegraphics[width=0.46\columnwidth]{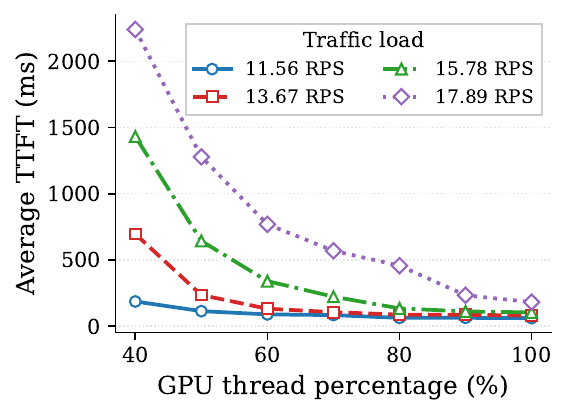}
        \label{fig:subfig2}
    }
\caption{Average TTFT of Yi-34B (vLLM, the largest model in our
pool) versus GPU thread allocation under varying traffic loads,
sharded across 2 and 4 GPUs. Latency varies by over an order of
magnitude across configurations at the same load, showing that
per-model latency depends jointly on resource allocation and
induced load.}
    \label{fig:thread}
\end{figure}

Existing routers~\cite{chenfrugalgpt, zooter, ongroutellm,
best-route, hybridllm, confidence, kasnavieh2026introlm} are
trained from prompts alone and treat each candidate model as if
it were a stateless API: a black box with fixed per-request cost
and latency, independent of deployment configuration and
workload~\cite{hybridllm, ongroutellm}. In real deployments,
however, models are not stateless APIs. GPU budgets are limited,
and when models are co-located on shared GPUs to enable
fine-grained allocation of compute and memory~\cite{muxserve},
each model's latency becomes a non-linear function of its
allocated resources and of the request load induced by the
routing policy itself (Figure~\ref{fig:thread}). The same model
can be the right choice for a prompt under one deployment and the
wrong choice under another, and a router optimized against
fixed latencies can produce assignments that are slow,
infeasible, or simply suboptimal once deployed. Routing and
deployment are therefore \emph{tightly coupled}: the resources
allocated to a model determine its latency, while the routing
policy determines its load and output quality.

In this work, we study routing and resource allocation as a single
joint problem. Given a GPU cluster, a pool of models,
and a latency service-level objective (SLO), our goal is to
determine both (i) the per-model resource allocation, including GPU thread
share, memory, and tensor parallelism level, and (ii) a routing policy
that assigns prompts to models, so as to maximize output quality
subject to the SLO. We propose \textsc{RouterWise}, which searches
over feasible deployment setups and, for each, computes the best
routing policy under the SLO. The routing subproblem is solved via
a \emph{dual-price} formulation that reduces score-maximizing
routing under a load budget to simple per-prompt decisions, while
setup-specific latency models learned from system profiling
capture how each model behaves under different allocations and
loads. Our key contributions are:
(1) we identify a gap in existing routing methods: they assume fixed per-model latency, which fails under self-hosted deployment;
(2) we propose \textsc{RouterWise}, a deployment-aware router that conditions routing decisions on the model's deployment configuration via a dual-price formulation and profiled latency models;
(3) we show empirically that achievable routing quality varies by up to 87\% across deployment configurations on the same GPU cluster, demonstrating that routing and deployment cannot be optimized independently.

%% file: preliminaries.tex
\section{Preliminaries}
\subsection{GPU Sharing Model}
LLMs are typically deployed across multiple GPUs using tensor
parallelism (TP), which partitions model weights into shards, one
per GPU~\cite{xiang2025aegaeon}. The simplest strategy is to
dedicate a full GPU to each shard (Fig.~\ref{fig:colocation}(a)),
but in multi-model serving this can leave some GPUs underutilized
while others become bottlenecks, since different models face different request loads and resource demands. We therefore allow
multiple shards to share a GPU, enabling fine-grained allocation
of compute and memory across co-located models and shifting
resources toward those with higher load.

We enable compute sharing via NVIDIA MPS~\cite{nvidia_mps}, which
schedules CUDA kernels from co-located shards concurrently on the
same GPU and lets us cap the fraction of compute resources (active
threads) each shard can use. GPU memory is partitioned explicitly
across co-located models. Figure~\ref{fig:thread} shows that
increasing the thread cap reduces TTFT under various loads,
indicating that latency is highly sensitive to the compute share
assigned to each model.

\begin{figure}[t]
\includegraphics[width=\columnwidth]{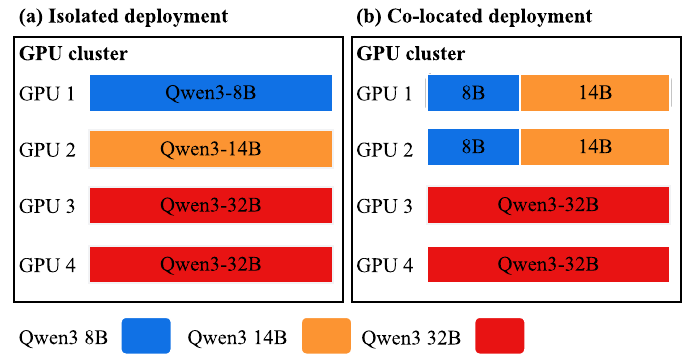}
\caption{Resource allocation strategies in multi-model LLM serving. Models can be deployed in isolation or co-located on shared GPUs. Co-location enables flexible compute and memory resource sharing and improves utilization.}
\label{fig:colocation}
\end{figure}

\subsection{Routing Model}
\label{subsection:routing_model}
We assume access to a lightweight router that, for each input
prompt $x$, outputs a score vector
\[
s(x) = (s_1(x), \dots, s_M(x)) \in [0,1]^M,
\]
where $s_i(x)$ is the predicted output quality of model $i$ on
$x$. We instantiate the router as a DeBERTa-v3-Large
encoder~\cite{he2020deberta} with a multi-head regression layer on
the \texttt{[CLS]} representation, one head per candidate
model, trained on the RouterBench~\cite{hu2024routerbench}. The
router supplies the per-prompt utility estimates used by our
framework, but is not itself the focus of this work; our
formulation only assumes access to per-model quality estimates
and is agnostic to the router architecture.

Figure~\ref{fig:scores} shows the predicted score distributions on
RouterBench. Larger models tend to score higher on average, but
the distributions overlap substantially: a non-trivial fraction of
prompts can be served effectively by smaller models. This
heterogeneity is what makes routing worthwhile---simpler prompts
can be handled by smaller models, with larger models reserved for
prompts that genuinely benefit from their capability. We use ``score'' and ``output quality'' interchangeably
throughout: $s_i(x)$ is the router's estimate of model $i$'s
output quality on $x$, calibrated against RouterBench's quality
labels.

\subsection{System Setup}
We consider a multi-model LLM serving system on a GPU cluster with
$G$ GPUs and a pool of $M$ models, indexed by $i \in \{1, \dots,
M\}$, serving an incoming workload at global arrival rate
$\lambda$ (requests per second). Each model $i$ is assigned a
\emph{setup}
\[
\theta_i = (\mathcal{G}_i, \mathrm{tp}_i, \rho_i) \in \Theta_i,
\]
where $\mathcal{G}_i \subseteq \{1, \dots, G\}$ is the set of
assigned GPUs, $\mathrm{tp}_i$ is the tensor parallelism level,
and $\rho_i \in (0, 1]$ is the fraction of GPU compute (thread
percentage) allocated to model $i$. GPU memory for each shard is
determined by profiling (see $m_i(\mathrm{tp}_i)$ below). A \emph{system setup}
$\Theta = (\theta_1, \dots, \theta_M) \in \mathcal{S}$ collects
the per-model setups, where $\mathcal{S}$ contains all
configurations that satisfy deployability constraints: GPU memory
capacity, tensor-parallel placement, and the requirement that all
shards fit on the available GPUs. For each model $i$ and parallelism
level $\mathrm{tp}_i$, we obtain $m_i(\mathrm{tp}_i)$, the minimum
memory fraction required by one shard, from profiling; this is
used to check shard placement feasibility.

\begin{figure}[t]
\centering
\includegraphics[width=0.8\columnwidth]{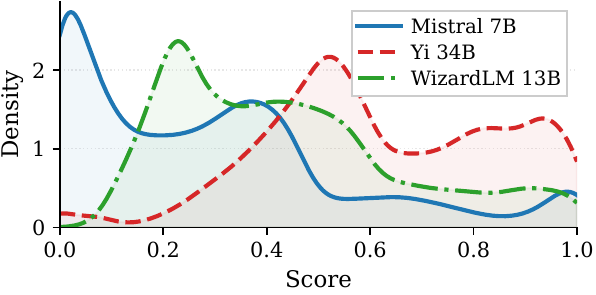}
\caption{Predicted routing scores on RouterBench. Larger models
score higher on average, but the distributions overlap
substantially, many prompts can be served effectively by smaller
models.}
\label{fig:scores}
\end{figure}

A routing policy $\pi: \mathcal{X} \rightarrow \{1, \dots, M\}$
maps each prompt $x$, drawn from a workload distribution
$\mathcal{D}$, to a deployed model. Our goal is to jointly choose
a system setup and a routing policy that maximize output quality
subject to a latency target. The two decisions are interdependent: a model's setup determines its latency and throughput
under load, while the routing policy determines the load each
model receives.

%% file: problem_formulation.tex
\section{Problem Formulation}
\label{sec:framework}
We formulate the joint resource-allocation-and-routing problem.
Let $\Theta \in \mathcal{S}$ denote a feasible system setup. A
routing policy makes prompt-level decisions, but over a workload
distribution it induces an aggregate traffic split that we capture
by a routing fraction vector
\begin{equation}
\label{eq:simplex}
w \in \Delta_M \triangleq \left\{ w \in \mathbb{R}_{\ge 0}^M \;\middle|\; \sum_{i=1}^M w_i = 1 \right\},
\end{equation}
where $w_i$ is the fraction of prompts routed to model $i$. This
workload-level view lets us characterize routing through two
quantities: the maximum achievable average score under $w$, and
the resulting latency under setup $\Theta$.

\paragraph{Maximum achievable score.}
Let $\mathcal{X} = \{x_1, \dots, x_N\}$ be a finite sample from
the workload distribution and $s_i(x_j) \in [0,1]$ the predicted
score of model $i$ on prompt $x_j$. Given $w$, the maximum average
score achievable subject to model $i$ serving a fraction $w_i$ of
the workload is
\begin{equation}
\label{eq:score_primal}
\hat{S}(w) = \frac{1}{N} \max_{\{z_{ji}\}}
\sum_{j=1}^{N}\sum_{i=1}^{M} z_{ji}\, s_i(x_j),
\end{equation}
subject to
\begin{align}
\label{eq:score_count}
z_{ji} \in \{0,1\}, \quad
\sum_{i=1}^{M} z_{ji} = 1 \;\forall j, \quad
\sum_{j=1}^{N} z_{ji} = c_i \;\forall i,
\end{align}
where $z_{ji} = 1$ indicates that prompt $x_j$ is assigned to
model $i$, and $c_i \approx N w_i$ is the target count for model
$i$ (with $\sum_i c_i = N$).

\paragraph{Latency under setup and routing.}
Under routing fractions $w$ and global arrival rate $\lambda$,
model $i$ receives load $\lambda_i = \lambda w_i$. We define
\begin{equation}
L(\Theta, w) = \sum_{i=1}^{M} w_i\, \ell_i(\theta_i, \lambda w_i),
\end{equation}
where $\ell_i(\theta_i, \lambda_i)$ is the average latency of
model $i$ under its setup $\theta_i$ and load $\lambda_i$. Depending on the
SLO, $\ell_i$ may denote end-to-end latency, time to first token
(TTFT), or time per output token (TPOT).

\begin{figure*}[t]
    \centering
    \subfloat[Mistral 7B\label{fig:ttft_mistral}]{
        \includegraphics[width=0.3\textwidth]{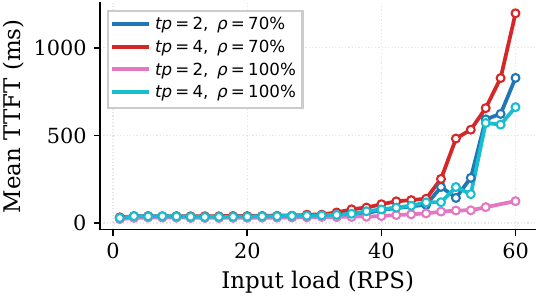}
    }
    \hfill
    \subfloat[WizardLM 13B\label{fig:ttft_vicuna}]{
        \includegraphics[width=0.3\textwidth]{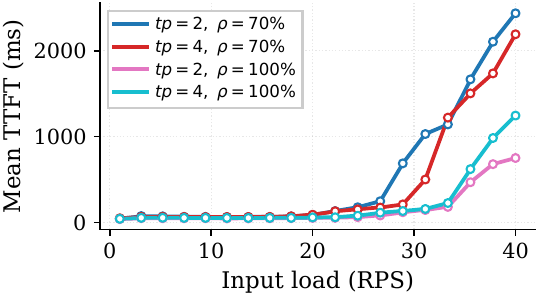}
    }
    \hfill
    \subfloat[Yi 34B\label{fig:ttft_yi}]{
        \includegraphics[width=0.3\textwidth]{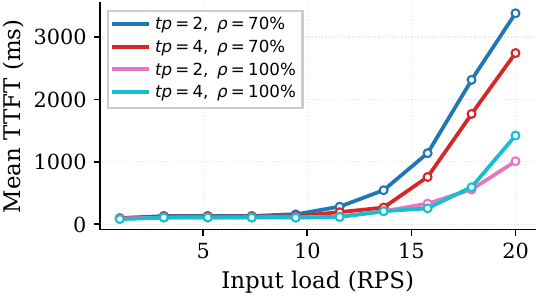}
    }
    \caption{Average TTFT versus input load across tensor
    parallelism levels and GPU compute shares ($\rho$) for different models in RouterBench obtained by profiling. The
    latency-load relationship varies substantially across both
    model size and configuration.}
    \label{fig:latency_profile}
\end{figure*}

\paragraph{Joint optimization.}
The joint problem is
\begin{equation}
\label{optimization_formula}
\max_{\Theta \in \mathcal{S},\, w \in \Delta_M} \hat{S}(w)
\qquad \text{s.t. } L(\Theta, w) \le \tau,
\end{equation}
where $\tau$ is the latency target. We focus on average-latency SLOs in this work; extending the
formulation to tail-latency targets (P95, P99) is a natural
direction for future work.

%% file: methodology.tex
\section{RouterWise}
\label{sec:solution}

This section presents \textsc{RouterWise} in four parts: a
dual-price formulation of the score function
(\S\ref{subsection:dual_score}), setup-specific latency models
from profiling (\S\ref{subsection:latency_model}), the
optimization procedure that combines them
(\S\ref{subsection:procedure}), and the deployment of the
resulting plan on the cluster (\S\ref{subsec:deployment}).

\subsection{Dual-Price Formulation of \texorpdfstring{$\hat{S}(w)$}{S-hat(w)}}
\label{subsection:dual_score}

\begin{proposition}[Dual-price characterization]
\label{prop:dual_score}
Let $\alpha_i \in \mathbb{R}$ be the dual variable for the count
constraint $\sum_j z_{ji} = c_i$ in
Eqs.~\eqref{eq:score_primal}--\eqref{eq:score_count}. Dualizing
these constraints yields $\min_{\alpha \in \mathbb{R}^M} g(\alpha)$, with
\begin{equation}
\label{dual_obj}
g(\alpha) = \frac{1}{N}\!\left[\sum_{j=1}^{N}\max_{i}\bigl(s_i(x_j)-\alpha_i\bigr) + \sum_{i=1}^{M}\alpha_i c_i\right]\!.
\end{equation}
For any fixed $\alpha$, the inner maximization decomposes across
prompts and is solved by assigning each $x_j$ to
\begin{equation}
\label{eq:assignment_rule}
i^*(x_j;\alpha) = \arg\max_{i \in \{1,\dots,M\}}\bigl(s_i(x_j)-\alpha_i\bigr).
\end{equation}
\end{proposition}

\noindent\textit{Proof.} See Appendix~\ref{appendix:dual_score}. \hfill $\square$

The $\alpha_i$ admit a natural interpretation as model-specific
\emph{prices}: increasing $\alpha_i$ makes model $i$ less
attractive to all prompts; decreasing it makes the model more
attractive. The dual thus reduces the constrained assignment over
$N$ prompts to $N$ independent per-prompt argmax decisions
governed by these prices.
\begin{proposition}[Convexity]
\label{prop:dual_convex}
$g(\alpha)$ is convex in $\alpha \in \mathbb{R}^M$.
\end{proposition}
\noindent\textit{Proof.} See Appendix~\ref{appendix:dual_convex}. \hfill $\square$

By Proposition~\ref{prop:dual_convex}, $g$ can be minimized via
subgradient descent. Let
$n_i(\alpha) = \sum_j \mathbf{1}[i^*(x_j;\alpha) = i]$
denote the prompt count assigned to model $i$ under $\alpha$. A
subgradient of $g$ with respect to $\alpha_i$ is
$(c_i - n_i(\alpha))/N$, giving the update
\begin{equation}
\alpha_i^{(t+1)} = \alpha_i^{(t)} + \eta_t \frac{n_i(\alpha^{(t)}) - c_i}{N}.
\end{equation}
Intuitively, over-assigned models have their price raised, and
under-assigned models have it lowered.

\subsection{Setup-Specific Latency Models}
\label{subsection:latency_model}

Latency depends not only on input load $\lambda$ but also on the
serving configuration, tensor parallelism ($\mathrm{tp}$) and GPU
compute share ($\rho$), and the dependence is highly
non-uniform. Figure~\ref{fig:latency_profile} shows that the
latency-load relationship varies substantially across both
configurations and model families: for Mistral 7B, for instance,
$\mathrm{tp}=2$ yields lower TTFT than $\mathrm{tp}=4$ over a wide
range of loads, while the opposite holds for other models. A
single unified latency predictor is therefore unlikely to
generalize accurately across this space.

\textsc{RouterWise} instead adopts a setup-specific strategy. We
discretize the setup space (in our experiments, $\mathrm{tp} \in
\{1,2,4\}$ and $\rho$ on a discrete grid), profile each
(model, setup) pair under multiple input loads, and approximate
$\ell_i(\theta_i, \lambda_i)$ as a piecewise-linear interpolation
between profiled points. Given a system setup $\Theta$, the
overall latency is then
\begin{equation}
\label{eq:latency_sum}
L(\Theta, w) = \sum_{i=1}^{M} w_i\, \ell_i(\theta_i, \lambda w_i).
\end{equation}

We assume $\ell_i$ depends only on $\theta_i$ and $\lambda_i$,
treating MPS thread caps and explicit memory partitioning as
sufficient to mitigate co-location interference. In practice,
runtime latency on shared clusters fluctuates due to residual
interference, workload burstiness, and infrastructure noise; we
therefore treat $\ell_i$ as a planning model rather than a
steady-state guarantee, and \textsc{RouterWise} relies on the
relative ranking of setups by latency, which we find is preserved
in deployment even when absolute values shift.

\subsection{Optimization Procedure}
\label{subsection:procedure}

\textsc{RouterWise} solves the joint problem in Eq.~(\ref{optimization_formula})
through a nested structure: an outer search over feasible system
setups, and for each setup an inner score-maximizing routing solver
that handles the latency constraint via Lagrangian relaxation \cite{lagrangian}. The
relaxation introduces a multiplier $\beta \ge 0$ that penalizes
latency violations; we bisect on $\beta$ to find the operating
point that meets the latency target $\tau$.

\paragraph{Stage 1: Routing fractions for fixed $(\Theta, \beta)$.}
For fixed $\Theta$ and $\beta \ge 0$, \textsc{RouterWise} solves
\begin{equation}
\max_{w \in \Delta_M} \mathcal{J}(\Theta, w; \beta) \triangleq \hat{S}(w) - \beta\bigl(L(\Theta, w) - \tau\bigr)
\end{equation}
via projected gradient ascent on the simplex. The score gradient
is given by the dual prices (Proposition~\ref{prop:dual_gradient}
below), and the latency gradient is computed from the slopes of
the piecewise-linear curves in Eq.~\eqref{eq:latency_sum}:
\begin{equation}
\frac{\partial L(\Theta, w)}{\partial w_i}
= \ell_i(\theta_i, \lambda w_i) + \lambda w_i \frac{\partial \ell_i}{\partial \lambda_i}\bigg|_{\lambda_i = \lambda w_i}.
\end{equation}
Combining the two,
$\nabla_w \mathcal{J} = \alpha^\star - \beta \nabla_w L(\Theta, w)$,
and the update is
$w^{(t+1)} = \Pi_{\Delta_M}\!\bigl(w^{(t)} + \eta_t \nabla_w \mathcal{J}\bigr)$,
where $\Pi_{\Delta_M}$ is the simplex
projection~\cite{duchi2008efficient}. This is summarized in
Algorithm~\ref{alg:optimize_fractions}.

\begin{proposition}[Score gradient via dual prices]
\label{prop:dual_gradient}
Under the continuous approximation $c_i = N w_i$, if $\hat{S}(w)$
is differentiable at $w$ and $\alpha^\star$ is an optimal dual
solution, then $\nabla_w \hat{S}(w) = \alpha^\star$.
\end{proposition}
\noindent\textit{Proof.} See Appendix~\ref{appendix:dual_gradient}. \hfill $\square$

\begin{algorithm}[t]
\small
\caption{\textsc{OptimizeFractions}$(\Theta, \beta)$}
\label{alg:optimize_fractions}
\begin{algorithmic}[1]
\Require Setup $\Theta$, multiplier $\beta$, initial $w^{(0)} \in \Delta_M$, step sizes $\{\eta_t\}$, iterations $T$
\For{$t = 0, \dots, T-1$}
    \State Solve dual for $\hat{S}(w)$ to obtain $\alpha^\star$
    \State $g \gets \alpha^\star - \beta \nabla_w L(\Theta, w)$
    \State $w \gets \Pi_{\Delta_M}(w + \eta_t g)$
\EndFor
\State \Return $w$
\end{algorithmic}
\end{algorithm}

\paragraph{Stage 2: Bisection over $\beta$.}
Different values of $\beta$ trace out the score-latency frontier:
smaller $\beta$ favors score, larger $\beta$ favors lower
latency. \textsc{RouterWise} bisects on $\beta$ to find the
operating point that meets the latency target $\tau$
(Algorithm~\ref{alg:optimize_beta}).

\begin{algorithm}[t]
\small
\caption{\textsc{OptimizeBeta}$(\Theta)$}
\label{alg:optimize_beta}
\begin{algorithmic}[1]
\Require Setup $\Theta$, target $\tau$, bounds $\beta_{\min}, \beta_{\max}$, tolerance $\varepsilon$
\State $(w^\star, \beta^\star) \gets (\varnothing, \varnothing)$
\While{$\beta_{\max} - \beta_{\min} > \varepsilon$}
    \State $\beta \gets (\beta_{\min} + \beta_{\max})/2$
    \State $w \gets \textsc{OptimizeFractions}(\Theta, \beta)$
    \If{$L(\Theta, w) > \tau$}
        \State $\beta_{\min} \gets \beta$
    \Else
        \State $\beta_{\max} \gets \beta$; $(w^\star, \beta^\star) \gets (w, \beta)$
    \EndIf
\EndWhile
\State \Return $(w^\star, \beta^\star)$
\end{algorithmic}
\end{algorithm}

\paragraph{Stage 3: Setup search.}
We discretize each model's setup over candidate tensor parallelism
levels and GPU thread percentages, and form $\mathcal{S}$ as the
Cartesian product of these per-model choices.
\textsc{RouterWise} enumerates candidates from $\mathcal{S}$,
filtering each through two checks. We additionally introduce a utilization threshold $\rho_{\min}
\in [0, 1]$ that discards setups leaving the cluster
under-utilized. First, the total compute
demand $C(\Theta) = \sum_i \mathrm{tp}_i \rho_i$ must satisfy
$G\rho_{\min} \le C(\Theta) \le G$, discarding setups that
under-utilize or over-subscribe the cluster. Second, the induced
shards must be placeable on $G$ GPUs under per-shard memory
$m_i(\mathrm{tp}_i)$, with no two shards of the same model
co-located; we check this via first-fit decreasing. For each
surviving setup, \textsc{RouterWise} runs \textsc{OptimizeBeta}
and retains the setup with the highest achievable score
(Algorithm~\ref{alg:select_setup}).

\subsection{Deployment}
\label{subsec:deployment}
Because \textsc{RouterWise}'s optimization is driven by profiled
latency curves, while runtime latency can fluctuate due to network
congestion, neighboring workloads on shared clusters, and other
sources of noise, we treat the optimization output as a deployment
plan rather than a steady-state guarantee. Given the
optimization output $(\Theta^\star, w^\star)$, each model $i$ is
launched as a separate vLLM process with tensor parallelism
$\mathrm{tp}_i$ across its assigned GPUs $\mathcal{G}_i$, with
its compute share capped at $\rho_i$ via NVIDIA MPS
(\texttt{CUDA\_MPS\_ACTIVE\_THREAD\_PERCENTAGE}). The router is
deployed alongside the model processes: at request time, it
computes scores $s(x)$ for the incoming prompt $x$, the dual-price
rule $i^\star(x; \alpha^\star) = \arg\max_i (s_i(x) - \alpha_i^\star)$
selects a model, and the prompt is forwarded to the corresponding
vLLM process for generation.

\paragraph{Complexity.}
The cost of \textsc{RouterWise} is dominated by the setup search.
Let $K_i$ be the number of discretized configuration choices for
model $i$; naive enumeration of $\mathcal{S}$ scales as
$\prod_i K_i$, which grows quickly with $M$. In practice, the
compute-budget and deployability checks prune most candidates:
let $\widetilde{\mathcal{S}} \subseteq \mathcal{S}$ denote the
retained set. The total cost is
$O(|\widetilde{\mathcal{S}}| \cdot T_\beta \cdot
C_{\textsc{OptimizeFractions}})$, where $T_\beta$ is the number
of bisection steps in \textsc{OptimizeBeta}. As shown in
Figure~\ref{fig:retained_setups}, $|\widetilde{\mathcal{S}}|$
remains modest in our deployments, peaking around 700 for three
models and 17{,}000 for four models on up to 8 GPUs, making the
search tractable even though most prior routing work considers
only two-model settings~\cite{ongroutellm, hybridllm,
chenfrugalgpt, best-route}.

\begin{figure}[t]
    \centering
    \begin{minipage}[t]{0.48\linewidth}
        \centering
        \includegraphics[width=\linewidth]{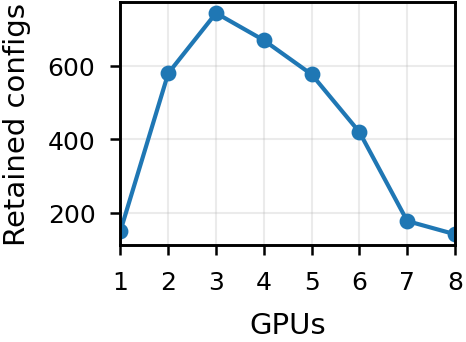}
        
        \vspace{2mm}
        \small (a) Three-model setting
    \end{minipage}
    \hfill
    \begin{minipage}[t]{0.48\linewidth}
        \centering
        \includegraphics[width=\linewidth]{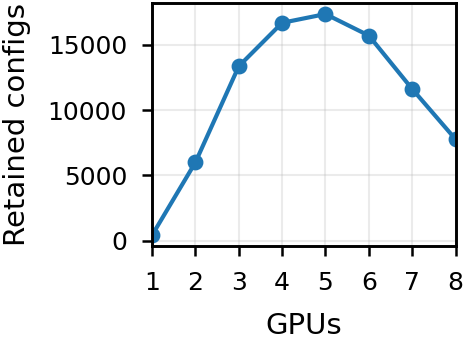}
        
        \vspace{2mm}
        \small (b) Four-model setting
    \end{minipage}
\caption{Number of retained setups as a function of cluster size,
for three- and four-model pools. Retained setups are those that
pass both the compute-utilization and deployability filters.}
    \label{fig:retained_setups}
\end{figure}

%% file: Evaluation.tex
\section{Evaluation}
\label{sec:evaluation}

We evaluate \textsc{RouterWise} along two complementary axes.
\textbf{Study 1} (\S\ref{subsec:study1}) asks whether resource
allocation matters for routing performance: holding the routing
policy fixed at \textsc{RouterWise}'s dual-price optimum, we
quantify how much achievable quality varies across feasible
deployment setups under the same SLO. \textbf{Study 2}
(\S\ref{subsec:study2}) evaluates  \textsc{RouterWise}  end-to-end against prior routing systems on a real GPU cluster. Both studies
share the experimental setup described next.

\begin{algorithm}[t]
\small
\caption{\textsc{SelectSetup}}
\label{alg:select_setup}
\begin{algorithmic}[1]
\Require Candidate setup space $\mathcal{S}$, GPUs $G$, latency target $\tau$, utilization threshold $\rho_{\min}$
\Ensure Best setup-routing pair $(\Theta^\star, w^\star)$
\State $(\Theta^\star, w^\star) \gets (\varnothing, \varnothing)$;\; $S^\star \gets -\infty$
\For{$\Theta = (\theta_1, \dots, \theta_M) \in \mathcal{S}$}
    \State $C(\Theta) \gets \sum_{i=1}^{M} \theta_i.\mathrm{tp}\,\theta_i.\rho$ \Comment{compute demand}
    \If{$C(\Theta) < G\,\rho_{\min}$ \textbf{or} $C(\Theta) > G$}
        \State \textbf{continue} \Comment{under-utilized or over-budget}
    \EndIf
\State $Q(\Theta) \gets [\,]$
\For{$i = 1, \dots, M$}
    \For{$r = 1, \dots, \theta_i.\mathrm{tp}$}
        \State append $m_i(\theta_i.\mathrm{tp})$ to $Q(\Theta)$
    \EndFor
\EndFor
    \If{\textbf{not} \textsc{FFD-Feasible}$(Q(\Theta), G)$}
        \State \textbf{continue} \Comment{shards do not fit}
    \EndIf
    \State $(w, \beta) \gets \textsc{OptimizeBeta}(\Theta)$
    \If{$L(\Theta, w) \le \tau$ \textbf{and} $\hat{S}(w) > S^\star$}
        \State $S^\star \gets \hat{S}(w)$;\; $(\Theta^\star, w^\star) \gets (\Theta, w)$
    \EndIf
\EndFor
\State \Return $(\Theta^\star, w^\star)$
\end{algorithmic}
\end{algorithm}

\subsection{Experimental Setup}
\label{subsec:exp_setup}

\paragraph{Cluster.} We evaluate \textsc{RouterWise} on GPU
clusters with $G \in \{4, 8\}$ NVIDIA A100 80GB PCIe GPUs, chosen
to study the method under both moderate and larger deployment
budgets.

\paragraph{Models and dataset.} While most prior routing work
focuses on routing between two
models~\cite{ongroutellm, zooter, hybridllm, automix, chenfrugalgpt},
we consider a more general multi-model setting with three models
spanning a range of sizes and capabilities:
\texttt{mistral-7b-chat}~\cite{Mistral3_2025},
\texttt{WizardLM-13B-V1.2}~\cite{xu2023wizardlm}, and
\texttt{Yi-34B-Chat}~\cite{young2024yi}. These models are drawn
from the RouterBench dataset~\cite{hu2024routerbench}, which
provides 36{,}000 prompts paired with model outputs and quality
scores derived from benchmark correctness criteria. Each model is
served with vLLM~\cite{vllm}, which is also used to construct the
setup-specific latency profiles.

    \begin{figure*}[t]
    \centering

    \subfloat[$G=4$: Maximum achievable score in different setups with 4 GPUs.%
    \label{fig:score_latency_g4}]{
        \begin{tabular}{cccc}
            \includegraphics[width=0.235\textwidth]{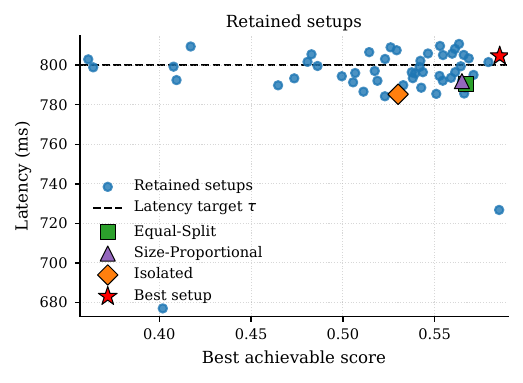 } &
            \includegraphics[width=0.235\textwidth]{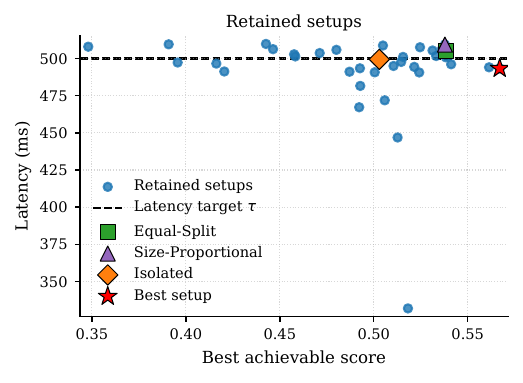} &
            \includegraphics[width=0.235\textwidth]{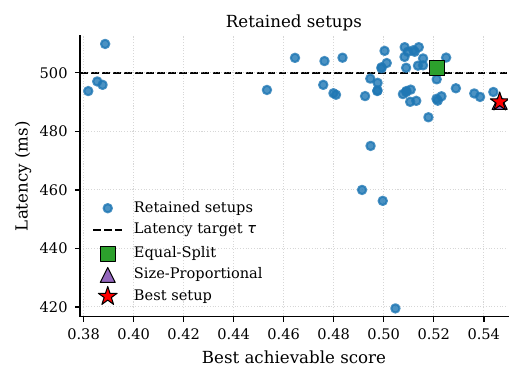} &
            \includegraphics[width=0.235\textwidth]{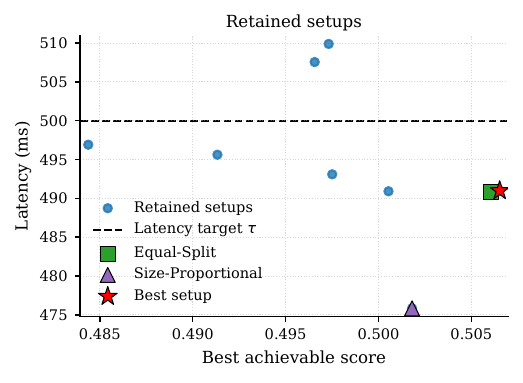} \\[-0.2em]
            {\footnotesize $\lambda=50$, $\tau=800$} &
            {\footnotesize $\lambda=60$, $\tau=500$} &
            {\footnotesize $\lambda=70$, $\tau=500$} &
            {\footnotesize $\lambda=80$, $\tau=500$}
        \end{tabular}
    }

    \vspace{0.6em}

    \subfloat[$G=8$: Maximum achievable score in different setups with 8 GPUs.%
    \label{fig:score_latency_g8}]{
        \begin{tabular}{cccc}
            \includegraphics[width=0.235\textwidth]{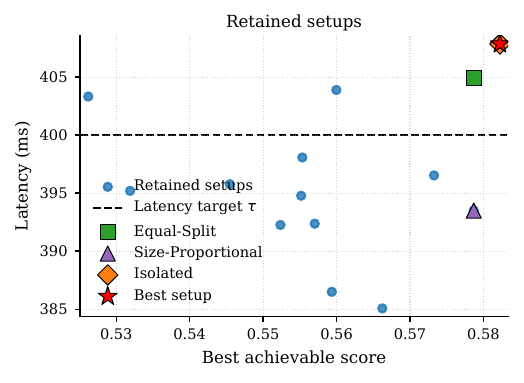} &
            \includegraphics[width=0.235\textwidth]{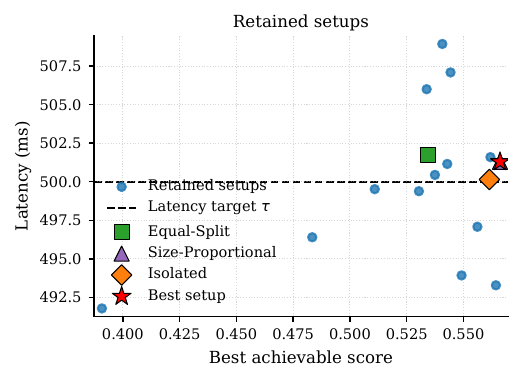} &
            \includegraphics[width=0.235\textwidth]{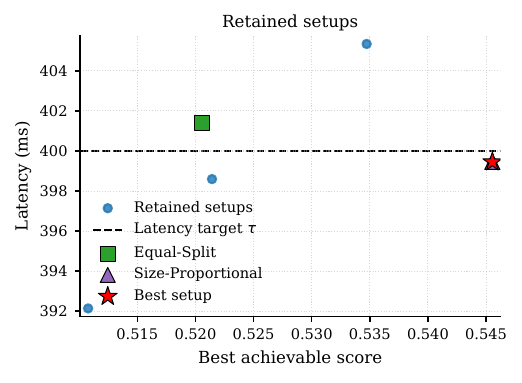} &
            \includegraphics[width=0.235\textwidth]{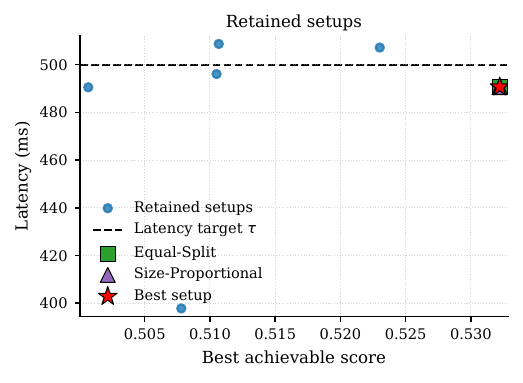} \\[-0.2em]
            {\footnotesize $\lambda=50$, $\tau=400$} &
            {\footnotesize $\lambda=60$, $\tau=500$} &
            {\footnotesize $\lambda=70$, $\tau=400$} &
            {\footnotesize $\lambda=80$, $\tau=500$}
        \end{tabular}
    }

    \caption{Score-latency scatter plots of retained setups under two settings of $G$, where retained setups denote deployable setups that also pass the compute-utilization filter. Row (a) shows the case $G=4$, and row (b) shows the case $G=8$.}
    \label{fig:score_latency_grid}
\end{figure*}
\paragraph{Router.} Following \S\ref{subsection:routing_model},
we train a multi-head DeBERTa-v3-Large~\cite{he2020deberta}
router on RouterBench, with one regression head per candidate
model. At inference, the router emits a score $s_i(x)$ per model
for each input $x$; these scores feed the dual-price assignment
rule $s_i(x) - \alpha_i$.

\paragraph{Setup space.} We discretize each model's tensor
parallelism level over $\{1, 2, 4\}$ and its GPU thread
percentage over $\{10, 20, \dots, 100\}$. Shards of the same
model may not be co-located on the same GPU.
\textsc{RouterWise} enumerates the Cartesian product of these
per-model choices and applies the feasibility checks from
\S\ref{subsection:procedure}.

Full hyperparameter values are listed in
Appendix~\ref{app:hyperparameters}.
\subsection{Study 1: Achievable Performance Across Different Setups}
\label{subsec:study1}

We first study how the deployment configuration alone shapes the
quality--latency trade-off available to a router. Holding the
routing policy fixed at \textsc{RouterWise}'s dual-price optimum,
we vary only the deployment setup and measure the achievable
score--latency operating point for each. For each
$(G, \lambda, \tau)$, we run the routing solver on every retained
setup (deployable, passes the compute-utilization filter) and
record the achieved score and latency. The resulting points
reveal the range of routing quality achievable under the same
SLO across different deployments, and thus how much the choice
of deployment, independent of the router, constrains what routing
can achieve.

\paragraph{Utilization filter.} We set $\rho_{\min} = 1$
throughout this study to avoid underutilizing the compute resources. Under NVIDIA MPS, the thread percentage
assigned to a process caps the fraction of GPU compute that the
corresponding shard can use; if the aggregate cap on a GPU is
below 100\%, part of its compute cannot be utilized concurrently
even when workload is available. Setting $\rho_{\min} = 1$
excludes such trivially under-utilized setups while preventing resource contention, so that any
observed score differences reflect \emph{how} compute and memory
are distributed across models, not whether compute is left idle.

\paragraph{Resource-allocation baselines.} Existing routing methods
treat per-model latency as fixed and therefore do not address
setup selection, so no prior routing method provides a direct
baseline for this study. We instead compare against three
allocation strategies, all instantiated on the same discretized
setup space as \textsc{RouterWise}:

\textbf{Equal-Split.} Each model is assigned a target compute
budget of $G/M$; we select the feasible setup whose per-model
allocations best match this equal-share target.

\textbf{Size-Proportional.} Each model $i$ is assigned a target
budget proportional to its parameter count $p_i$, i.e.,
$G p_i / \sum_k p_k$; we again select the closest feasible setup.

\textbf{Isolated.} Co-location is disabled and shards are placed
on dedicated GPUs whenever feasible, reflecting a conventional
deployment strategy that cannot benefit from fine-grained
sharing.

\paragraph{Results.}
Figure~\ref{fig:score_latency_grid} shows the score--latency
operating points of all retained setups, with the three
baselines and \textsc{RouterWise}'s selected setup highlighted.
A clear trend emerges across all $(G, \lambda, \tau)$ settings:
retained setups span a wide range of achievable scores under the
same SLO, with relative scores ranging from $0.30$ to $0.56$, a
spread of up to $87\%$ on the same cluster. This is despite
every plotted setup passing the utilization filter, so the
variation cannot be attributed to idle GPUs; it reflects
\emph{how} the available compute and memory are distributed
across models. Resource allocation is therefore not a deployment
detail but a first-order factor in routing performance, and the
three fixed-rule baselines fall short of the setup selected by
\textsc{RouterWise} in all settings we evaluate.

Comparing $G=4$ and $G=8$, we observe higher achievable scores at
$G=8$ as expected from the larger compute budget, and fewer
retained setups: with more resources available, the fraction of
discretized configurations that fully saturate the cluster under
$\rho_{\min}=1$ shrinks, leaving a smaller feasible set to search
and making the setup search correspondingly faster.

\begin{table}[t]
\centering
\small
\setlength{\tabcolsep}{4pt}
\begin{tabular}{llccc}
\toprule
& & \multicolumn{3}{c}{Score / Average TTFT (ms)} \\
\cmidrule(lr){3-5}
$\lambda$ & $\tau$ (ms) & \textsc{RouterWise} & FrugalGPT & RouteLLM \\
\midrule
50 & 800 & \textbf{0.59} / 860 & 0.59 / 1500 & 0.58 / 820 \\
60 & 500 & \textbf{0.57} / 550 & 0.57 / 2700 & 0.57 / 1100 \\
70 & 500 & \textbf{0.56 } / 565 & 0.56 / 12300 & 0.54 / 2500 \\
80 & 500 & \textbf{0.51} / 600 & 0.51 / 24200 & 0.51 / 7800 \\
\bottomrule
\end{tabular}
\caption{End-to-end deployment results on 4 A100 GPUs across
four offered loads ($\lambda$, RPS) under a shared latency target
$\tau$.}
\label{tab:study2_g4}
\end{table}

\subsection{Study 2: End-to-End Deployment}
\label{subsec:study2}

We evaluate \textsc{RouterWise} end-to-end against two
representative routing baselines on a real GPU cluster.

\paragraph{Baselines.}
We compare against two widely cited routers covering the two
dominant routing paradigms: \textbf{FrugalGPT}~\cite{chenfrugalgpt},
which cascades prompts through models from cheapest to most
expensive and escalates only when a learned quality score is
insufficient, and \textbf{RouteLLM}~\cite{ongroutellm}, which
performs pre-generation routing by predicting the best model
per prompt. Neither addresses deployment configuration; we
therefore deploy each baseline in the conventional manner with
models isolated on dedicated GPUs, while \textsc{RouterWise} uses
the configuration chosen by its joint optimization. To isolate
the deployment contribution from the routing layer, we tune each
baseline's quality--cost parameter to match
\textsc{RouterWise}'s score (within $0.01$) at each
$(\lambda, \tau)$; the comparison thus reduces to TTFT at
equivalent routing quality.

\paragraph{Setup.}
We deploy on $G = 4$ A100 80GB GPUs and sweep offered load
$\lambda \in \{50, 60, 70, 80\}$ RPS under latency targets
$\tau \in \{500, 800\}$\,ms, reporting mean score and measured
average TTFT.

\paragraph{Results.}
Table~\ref{tab:study2_g4} shows the outcomes. Because scores are
matched by construction, the meaningful comparison is in TTFT.
\textsc{RouterWise} stays close to the SLO across all operating
points, meeting the target at $\lambda \in \{50, 60, 70\}$ and
overshooting by only $100$\,ms at $\lambda = 80$. Both baselines
overshoot far more, and the gap widens sharply with load.

\paragraph{Analysis.}
The two baselines fail for different but related reasons. 
FrugalGPT's cascade structure sends every prompt to the cheapest
model first and escalates a large fraction to larger models,
inverting the natural relationship between model size and serving
capacity; since the deployment is isolated and cannot adapt to
this skewed load, the cheapest model's queue saturates almost
immediately. RouteLLM avoids this cascade overhead by dispatching
each prompt directly to a single model, but isolated deployment
still leaves the most-used model as a bottleneck. In both cases,
the routing decisions themselves are of comparable quality to
\textsc{RouterWise}, while the latency gap is mainly attributable to deployment. This isolates the central claim of our paper:
when routing quality is held fixed, deployment-aware allocation
is what delivers SLO compliance.

%% file: conclusion.tex
\section{Conclusion}
\label{sec:conclusion}

We have argued that prompt routing cannot be studied in isolation from
deployment: when models share a GPU cluster, per-model latency is
not a fixed property but a joint function of resource allocation
and routing-induced load. Existing routers, trained from prompts
alone, miss this dependence and produce assignments that are
slow or infeasible under realistic deployment. We proposed
\textsc{RouterWise}, a deployment-aware router that jointly
selects a deployment configuration and a routing policy via a
dual-price formulation over profiled latency models.
Experiments on three open-source LLMs show that achievable
routing quality varies by up to $87\%$ across deployment
configurations on the same cluster. The optimal model for a prompt depends not only
on the prompt itself, but also on the deployment environment.

%% file: limitation.tex
\section{Acknowledgements}
We used OpenAI ChatGPT to assist with proofreading and improving the clarity and presentation of the manuscript.

\section{Limitations}
\label{sec:limitations}

The primary objective of this work is to demonstrate that
resource allocation and prompt routing are tightly coupled, and
to propose a formal framework for finding an optimal joint
solution under a latency SLO. Establishing this coupling, rather
than building a production-ready serving system, motivates
several assumptions that real deployments will not satisfy.

First, we assume the prompt distribution is fixed and known: the
router is trained on a representative sample, and the dual-price
optimization is computed against that sample. In practice, prompt
distributions shift over time due to workload evolution, domain
changes, and user behavior. Adapting routing and allocation
decisions online as the distribution drifts is an open question.

Second, our latency models are constructed offline from
profiling and treat each model in isolation. Runtime latency on shared clusters fluctuates due to co-location interference, workload burstiness, and infrastructure noise, and we observe modest gaps between profiled and deployed TTFT in our experiments (\S\ref{subsec:study2}). Although
\textsc{RouterWise} relies on the relative ranking of setups
rather than absolute predictions, incorporating online
calibration would tighten the planning loop and improve
robustness.

Third, we optimize against a fixed input load $\lambda$. Real
serving systems face fluctuating arrival rates, and a setup
selected for one load may be suboptimal under another.
Extending the formulation to multiple load regimes, or to
adaptive resource reallocation as load changes, is a natural
direction.

The assumptions above scope the contribution of this work
without limiting its central claim. Building robust serving
systems that operate well under distribution shift, latency
noise, and load variation is a substantial research agenda in
its own right; the framework presented here provides a
principled foundation for that work, and the central
observation, that the optimal model for a prompt depends on the
deployment environment, holds regardless of which extensions
are layered on top.

%% file: appendix.tex
\appendix

\section{Proof of Proposition~\ref{prop:dual_score}}
\label{appendix:dual_score}

In this appendix, we derive the dual-price formulation of the primal
score-maximization problem.

Recall the primal problem:
\begin{equation}
\label{eq:appendix_score_primal}
\hat{S}(w)
=
\frac{1}{N}
\max_{\{z_{ji}\}}
\sum_{j=1}^{N}\sum_{i=1}^{M} z_{ji}\, s_i(x_j)
\end{equation}
subject to
\begin{align}
z_{ji} &\in \{0,1\}, && \forall j,i, \label{eq:appendix_score_binary}\\
\sum_{i=1}^{M} z_{ji} &= 1, && \forall j, \label{eq:appendix_score_assign}\\
\sum_{j=1}^{N} z_{ji} &= c_i, && \forall i. \label{eq:appendix_score_count}
\end{align}
Here, $z_{ji}=1$ indicates that prompt $x_j$ is assigned to model $i$, and
$c_i$ is the target number of prompts assigned to model $i$.

We derive the dual by dualizing only the target-count constraints
\eqref{eq:appendix_score_count}. Let $\alpha_i \in \mathbb{R}$ denote the dual
variable associated with the constraint
\[
\sum_{j=1}^{N} z_{ji}=c_i.
\]
The corresponding Lagrangian is
\begin{align}
\label{eq:appendix_lagrangian}
\mathcal{J}(Z,\alpha)
&= \frac{1}{N}\sum_{j=1}^{N}\sum_{i=1}^{M} z_{ji}\, s_i(x_j) \notag \\
&\quad + \frac{1}{N}\sum_{i=1}^{M}\alpha_i \left(c_i-\sum_{j=1}^{N} z_{ji}\right).
\end{align}

Rearranging \eqref{eq:appendix_lagrangian} gives
\begin{align}
\label{eq:appendix_lagrangian_rearranged}
\mathcal{J}(Z,\alpha)
= \frac{1}{N}\Bigg[ &\sum_{j=1}^{N}\sum_{i=1}^{M} z_{ji}\bigl(s_i(x_j)-\alpha_i\bigr) \notag \\
&+ \sum_{i=1}^{M}\alpha_i c_i \Bigg].
\end{align}

For fixed $\alpha$, the term $\sum_{i=1}^{M}\alpha_i c_i$ is constant with
respect to the assignment variables. Therefore, maximizing
$\mathcal{J}(Z,\alpha)$ over $Z$ reduces to maximizing
\[
\sum_{j=1}^{N}\sum_{i=1}^{M}
z_{ji}\bigl(s_i(x_j)-\alpha_i\bigr)
\]
subject to the remaining constraints
\[
z_{ji}\in\{0,1\}, \qquad
\sum_{i=1}^{M} z_{ji}=1 \quad \forall j.
\]

These constraints couple variables only within each prompt $j$. Hence, the
maximization decomposes across prompts. For a fixed prompt $x_j$, assigning it
to model $i$ contributes
\[
s_i(x_j)-\alpha_i
\]
to the objective. Since each prompt must be assigned to exactly one model, the
optimal choice is
\begin{equation}
\label{eq:appendix_adjusted_score_rule}
i^*(x_j;\alpha)
=
\arg\max_{i\in\{1,\dots,M\}}
\bigl(s_i(x_j)-\alpha_i\bigr).
\end{equation}

Substituting this optimal prompt-wise assignment into
\eqref{eq:appendix_lagrangian_rearranged} yields the dual function
\begin{align}
g(\alpha)
&= \max_{Z}\mathcal{J}(Z,\alpha) \nonumber\\
&= \frac{1}{N}\Bigg[
\sum_{j=1}^{N} \max_{i\in\{1,\dots,M\}}\bigl(s_i(x_j)-\alpha_i\bigr) \nonumber\\
&\quad\;\, + \sum_{i=1}^{M}\alpha_i c_i
\Bigg].
\label{eq:appendix_dual_function}
\end{align}

Therefore, the dual problem is
\begin{equation}
\label{eq:appendix_dual_problem}
\min_{\alpha\in\mathbb{R}^M} g(\alpha),
\end{equation}
where $g(\alpha)$ is given by \eqref{eq:appendix_dual_function}.

This proves the dual-price characterization in
Proposition~\ref{prop:dual_score}. In addition,
\eqref{eq:appendix_adjusted_score_rule} shows that, for fixed dual prices
$\alpha$, the maximizing inner assignment is obtained independently for each
prompt by selecting the model with the largest adjusted score
$s_i(x_j)-\alpha_i$.

\section{Proof of Proposition~\ref{prop:dual_convex}}
\label{appendix:dual_convex}

In this appendix, we prove that the dual objective
\[
g(\alpha)
=
\frac{1}{N}
\left[
\sum_{j=1}^{N}
\max_{i\in\{1,\dots,M\}}
\bigl(s_i(x_j)-\alpha_i\bigr)
+
\sum_{i=1}^{M}\alpha_i c_i
\right]
\]
is convex in $\alpha \in \mathbb{R}^M$.

For each prompt $x_j$, define
\[
h_j(\alpha)
=
\max_{i\in\{1,\dots,M\}}
\bigl(s_i(x_j)-\alpha_i\bigr).
\]
Then $g(\alpha)$ can be written as
\[
g(\alpha)
=
\frac{1}{N}
\left[
\sum_{j=1}^{N} h_j(\alpha)
+
\sum_{i=1}^{M}\alpha_i c_i
\right].
\]

We first show that each function $h_j(\alpha)$ is convex. For any fixed model
$i$, the function
\[
h_{j,i}(\alpha)=s_i(x_j)-\alpha_i
\]
is affine in $\alpha$, since it is a constant $s_i(x_j)$ plus a linear term
$-\alpha_i$. The function $h_j(\alpha)$ is the pointwise maximum of the finite
collection of affine functions $\{h_{j,i}(\alpha)\}_{i=1}^{M}$. A pointwise
maximum of affine functions is convex. Therefore, each $h_j(\alpha)$ is convex.

Next, the term
\[
\sum_{i=1}^{M}\alpha_i c_i
\]
is linear in $\alpha$, and every linear function is both convex and concave.
Therefore, this term is convex.

Finally, $g(\alpha)$ is obtained by summing the convex functions
$\{h_j(\alpha)\}_{j=1}^{N}$ and the linear term $\sum_{i=1}^{M}\alpha_i c_i$,
and then scaling the result by the positive constant $1/N$. Since nonnegative
weighted sums of convex functions remain convex, it follows that $g(\alpha)$ is
convex in $\alpha$.

This proves Proposition~\ref{prop:dual_convex}.

\section{Proof of Proposition~\ref{prop:dual_gradient}}
\label{appendix:dual_gradient}

Let
\[
V(c)
=
\frac{1}{N}
\max_{\{z_{ji}\}}
\sum_{j=1}^{N}\sum_{i=1}^{M} z_{ji}\, s_i(x_j)
\]
denote the optimal value of the count-constrained score-maximization problem as
a function of the target count vector $c=(c_1,\dots,c_M)$. By
Proposition~\ref{prop:dual_score}, $V(c)$ admits the dual representation
\[
V(c)=\min_{\alpha\in\mathbb{R}^M} g(\alpha;c),
\]
where
\[
g(\alpha;c)
=
\frac{1}{N}
\left[
\sum_{j=1}^{N}\max_i\bigl(s_i(x_j)-\alpha_i\bigr)
+
\sum_{i=1}^{M}\alpha_i c_i
\right].
\]

Assume that $V(c)$ is differentiable at the current $c$, and let
$\alpha^\star(c)$ be an optimal dual solution. Since the dependence of
$g(\alpha;c)$ on $c$ appears only through the linear term
\[
\frac{1}{N}\sum_{i=1}^{M}\alpha_i c_i,
\]
the envelope theorem gives
\[
\frac{\partial V(c)}{\partial c_i}
=
\frac{\partial g(\alpha;c)}{\partial c_i}\Big|_{\alpha=\alpha^\star}
=
\frac{\alpha_i^\star}{N}.
\]

Now define $\hat S(w)=V(c)$ under the continuous approximation $c_i=Nw_i$.
Applying the chain rule,
\[
\frac{\partial \hat S(w)}{\partial w_i}
=
\sum_{k=1}^{M}
\frac{\partial V(c)}{\partial c_k}
\frac{\partial c_k}{\partial w_i}.
\]
Because $c_k=Nw_k$,
\[
\frac{\partial c_k}{\partial w_i}=N\,\mathbf{1}[k=i].
\]
Therefore,
\[
\frac{\partial \hat S(w)}{\partial w_i}
=
\frac{\alpha_i^\star}{N}\cdot N
=
\alpha_i^\star.
\]

Hence,
\[
\nabla_w \hat S(w)=\alpha^\star.
\]
This proves the proposition.

\section{Hyperparameters}
\label{app:hyperparameters}

We document the hyperparameter values used in our experiments
(\S\ref{sec:evaluation}). All values were chosen on a held-out
subset of RouterBench and held fixed across all reported runs.

\paragraph{Setup search space (\S\ref{subsection:procedure}).}
Tensor parallelism levels are drawn from $\{1, 2, 4\}$ and GPU
thread percentages from $\{10, 20, \dots, 100\}$. These values are also used to profile different model latencies.

\paragraph{Feasibility check.}
The first-fit-decreasing placement uses a memory slack tolerance
of $0.02$ to absorb minor profiling jitter. The compute-utilization
filter requires the aggregate thread sum to reach the cluster's compute budget ($\rho_{\min} = 1$).

\paragraph{Fraction optimizer (Algorithm~\ref{alg:optimize_fractions}).}
Projected gradient ascent on the simplex runs for $200$ iterations
with constant step size $\eta_t = 0.1$ (the step size in the
update rule in \S\ref{subsection:procedure}), initial Lagrange
multiplier $\beta_0 = 0.1$, momentum $0$, early-stopping patience
$10$, and objective tolerance $10^{-6}$.
\paragraph{Dual-price subgradient (\S\ref{subsection:dual_score}).}
The dual problem is solved with $300$ subgradient iterations,
initial step size $\eta_0 = 10^{-4}$, count tolerance $1$, and a
small tie-breaking noise of $10^{-9}$ added to scores to avoid
plateaus.

\paragraph{$\beta$ bisection (Algorithm~\ref{alg:optimize_beta}).}
The Lagrange multiplier $\beta$ is searched via bisection over
$[0, 5]$ for $80$ outer steps with multiplicative step decay
$0.99$.

%% file: custom.bib
@misc{anthropic2024,
  author = {Anthropic},
  title = {{Claude}},
  year = {2025},
  howpublished = {\url{https://www.anthropic.com}}
}

@misc{OpenAI2025,
  author = {OpenAI},
  title = {{ChatGPT}},
  year = {2025},
  howpublished = {\url{https://openai.com}}
}

@article{yang2025qwen3,
  title = {{Qwen3} Technical Report},
  author = {Yang, An and Li, Anfeng and Yang, Baosong and Zhang, Beichen and Hui, Binyuan and Zheng, Bo and Yu, Bowen and Gao, Chang and Huang, Chengen and Lv, Chenxu and others},
  journal = {arXiv preprint arXiv:2505.09388},
  year = {2025},
  url = {https://doi.org/10.48550/arXiv.2505.09388}
}

@article{Llama4_2025,
  author = {{Meta AI}},
  title = {{The Llama 4 Herd}: The Beginning of a New Era of Natively Multimodal {AI} Innovation},
  journal = {Meta AI Blog},
  year = {2025},
  month = {April},
  url = {https://ai.meta.com/blog/llama-4-multimodal-intelligence/},
  note = {Technical Release and Model Documentation}
}

@misc{Mistral3_2025,
  author = {{Mistral AI Team}},
  title = {Introducing {Mistral 3}},
  howpublished = {Mistral AI Blog},
  year = {2025},
  month = {December},
  url = {https://mistral.ai/news/mistral-3},
  note = {Official Model Release and Technical Documentation}
}

@inproceedings{gllm,
  title={gLLM: Global Balanced Pipeline Parallelism Systems for Distributed LLMs Serving with Token Throttling},
  author={Guo, Tianyu and Zhang, Xianwei and Du, Jiangsu and Chen, Zhiguang and Xiao, Nong and Lu, Yutong},
  booktitle={Proceedings of the International Conference for High Performance Computing, Networking, Storage and Analysis},
  pages={1725--1741},
  year={2025},
  url={https://doi.org/10.1145/3712285.3759823}
}

@inproceedings{hetis,
  title={Hetis: Serving llms in heterogeneous gpu clusters with fine-grained and dynamic parallelism},
  author={Mo, Zizhao and Liao, Jianxiong and Xu, Huanle and Zhou, Zhi and Xu, Chengzhong},
  booktitle={Proceedings of the International Conference for High Performance Computing, Networking, Storage and Analysis},
  pages={1710--1724},
  year={2025},
  url={https://doi.org/10.1145/3712285.3759784}
}

@inproceedings{ongroutellm,
  title = {{RouteLLM}: Learning to Route {LLMs} from Preference Data},
  author = {Ong, Isaac and Almahairi, Amjad and Wu, Vincent and Chiang, Wei-Lin and Wu, Tianhao and Gonzalez, Joseph E. and Kadous, M. Waleed and Stoica, Ion},
  booktitle = {Proceedings of the International Conference on Learning Representations},
  year = {2025},
  url={https://doi.org/10.48550/arXiv.2406.18665}
}

@article{chenfrugalgpt,
  title = {{FrugalGPT}: How to Use Large Language Models While Reducing Cost and Improving Performance},
  author = {Chen, Lingjiao and Zaharia, Matei and Zou, James},
  journal = {Transactions on Machine Learning Research},
  year = {2023},
  url={https://doi.org/10.48550/arXiv.2305.05176}
}

@inproceedings{zooter,
  title = {Routing to the Expert: Efficient Reward-Guided Ensemble of Large Language Models},
  author = {Lu, Keming and Yuan, Hongyi and Lin, Runji and Lin, Junyang and Yuan, Zheng and Zhou, Chang and Zhou, Jingren},
  booktitle = {Proceedings of the 2024 Conference of the North American Chapter of the Association for Computational Linguistics: Human Language Technologies},
  pages = {1964--1974},
  year = {2024},
  url = {https://doi.org/10.18653/v1/2024.naacl-long.109}
}

@inproceedings{hybridllm,
  title = {{Hybrid LLM}: Cost-Efficient and Quality-Aware Query Routing},
  author = {Ding, Dujian and Mallick, Ankur and Wang, Chi and Sim, Robert and Mukherjee, Subhabrata and R{\"u}hle, Victor and Lakshmanan, Laks V. S. and Awadallah, Ahmed Hassan},
  booktitle = {Proceedings of the International Conference on Learning Representations},
  year = {2024},
  url={https://doi.org/10.48550/arXiv.2404.14618}
}

@inproceedings{best-route,
  title = {{BEST-Route}: Adaptive {LLM} Routing with Test-Time Optimal Compute},
  author = {Ding, Dujian and Mallick, Ankur and Zhang, Shaokun and Wang, Chi and Madrigal, Daniel and Garcia, Mirian Del Carmen Hipolito and Xia, Menglin and Lakshmanan, Laks V. S. and Wu, Qingyun and R{\"u}hle, Victor},
  booktitle = {Proceedings of the International Conference on Machine Learning},
  year = {2025},
  url = {https://doi.org/10.48550/arXiv.2506.22716}
}

@inproceedings{confidence,
  title = {Learning to Route {LLMs} with Confidence Tokens},
  author = {Chuang, Yu-Neng and Sarma, Prathusha Kameswara and Gopalan, Parikshit and Boccio, John and Bolouki, Sara and Hu, Xia and Zhou, Helen},
  booktitle = {Proceedings of the International Conference on Machine Learning},
  year = {2025},
  url = {https://doi.org/10.48550/arXiv.2410.13284}
}

@article{automix,
  title = {{AutoMix}: Automatically Mixing Language Models},
  author = {Aggarwal, Pranjal and Madaan, Aman and Anand, Ankit and Potharaju, Srividya Pranavi and Mishra, Swaroop and Zhou, Pei and Gupta, Aditya and Rajagopal, Dheeraj and Kappaganthu, Karthik and Yang, Yiming and others},
  journal = {Advances in Neural Information Processing Systems},
  year = {2023},
  url = {https://doi.org/10.48550/arXiv.2310.12963}
}

@article{he2020deberta,
  title = {{DeBERTa}: Decoding-Enhanced {BERT} with Disentangled Attention},
  author = {He, Pengcheng and Liu, Xiaodong and Gao, Jianfeng and Chen, Weizhu},
  journal = {arXiv preprint arXiv:2006.03654},
  year = {2020} , 
  url = {https://doi.org/10.48550/arXiv.2006.03654}
}

@article{kasnavieh2026introlm,
  title={IntroLM: Introspective Language Models via Prefilling-Time Self-Evaluation},
  author={Kasnavieh, Hossein Hosseini and Haffari, Gholamreza and Leckie, Chris and Toosi, Adel N},
  journal={arXiv preprint arXiv:2601.03511},
  year={2026},
  url={https://doi.org/10.48550/arXiv.2601.03511}
}

@inproceedings{xiang2025aegaeon,
  title={Aegaeon: Effective GPU pooling for concurrent LLM serving on the market},
  author={Xiang, Yuxing and Li, Xue and Qian, Kun and Yang, Yufan and Zhu, Diwen and Yu, Wenyuan and Zhai, Ennan and Liu, Xuanzhe and Jin, Xin and Zhou, Jingren},
  booktitle={Proceedings of the ACM SIGOPS 31st Symposium on Operating Systems Principles},
  pages={1030--1045},
  year={2025}
}

@manual{nvidia_mps,
  title        = {NVIDIA CUDA Multi-Process Service},
  author       = {{NVIDIA}},
  year         = {2025},
  note         = {NVIDIA Deployment Guide, Release r590},
  url          = {https://docs.nvidia.com/deploy/mps/index.html}
}

@inproceedings{
hu2024routerbench,
title={RouterBench: A Benchmark for Multi-{LLM} Routing System},
author={Qitian Jason Hu and Jacob Bieker and Xiuyu Li and Nan Jiang and Benjamin Keigwin and Gaurav Ranganath and Kurt Keutzer and Shriyash Kaustubh Upadhyay},
booktitle={Agentic Markets Workshop at ICML 2024},
year={2024},
url={https://doi.org/10.48550/arXiv.2403.12031}
}

@inproceedings{duchi2008efficient,
author = {Duchi, John and Shalev-Shwartz, Shai and Singer, Yoram and Chandra, Tushar},
title = {Efficient projections onto the l1-ball for learning in high dimensions},
year = {2008},
isbn = {9781605582054},
publisher = {Association for Computing Machinery},
address = {New York, NY, USA},
url = {https://doi.org/10.1145/1390156.1390191},
doi = {10.1145/1390156.1390191},
booktitle = {Proceedings of the 25th International Conference on Machine Learning},
pages = {272–279},
numpages = {8},
location = {Helsinki, Finland},
series = {ICML '08}
}

@article{xu2023wizardlm,
  title   = {WizardLM: Empowering Large Language Models to Follow Complex Instructions},
  author  = {Xu, Can and Sun, Qingfeng and Zheng, Kai and Geng, Xiubo and Zhao, Pengcheng and Feng, Chong and Tao, Dacheng},
  journal = {arXiv preprint arXiv:2304.12244},
  year    = {2023},
  doi     = {10.48550/arXiv.2304.12244},
  url     = {https://doi.org/10.48550/arXiv.2304.12244}
}

@article{young2024yi,
  title   = {Yi: Open Foundation Models by 01.AI},
  author  = {Young, Alex and Chen, Bei and Li, Biao and Cai, Chong and Cao, Daya and Ge, Guodong and Li, Heng and Lin, Hongwei and Ning, Nanyun and Qu, Jiahao and Rajasekaran, Sivakumar and others},
  journal = {arXiv preprint arXiv:2403.04652},
  year    = {2024},
  doi     = {10.48550/arXiv.2403.04652},
  url     = {https://doi.org/10.48550/arXiv.2403.04652}
}

@inproceedings{vllm,
author = {Kwon, Woosuk and Li, Zhuohan and Zhuang, Siyuan and Sheng, Ying and Zheng, Lianmin and Yu, Cody Hao and Gonzalez, Joseph and Zhang, Hao and Stoica, Ion},
title = {Efficient Memory Management for Large Language Model Serving with PagedAttention},
year = {2023},
isbn = {9798400702297},
publisher = {Association for Computing Machinery},
address = {New York, NY, USA},
url = {https://doi.org/10.1145/3600006.3613165},
doi = {10.1145/3600006.3613165},
booktitle = {Proceedings of the 29th Symposium on Operating Systems Principles},
pages = {611–626},
numpages = {16},
location = {Koblenz, Germany},
series = {SOSP '23}
}

@inproceedings{muxserve,
author = {Duan, Jiangfei and Lu, Runyu and Duanmu, Haojie and Li, Xiuhong and Zhang, Xingcheng and Lin, Dahua and Stoica, Ion and Zhang, Hao},
title = {MuxServe: flexible spatial-temporal multiplexing for multiple LLM serving},
year = {2024},
publisher = {JMLR.org},
booktitle = {Proceedings of the 41st International Conference on Machine Learning},
articleno = {473},
numpages = {13},
location = {Vienna, Austria},
series = {ICML'24}
}

@article{lagrangian,
author = {Rush, Alexander M. and Collins, Michael},
title = {A tutorial on dual decomposition and lagrangian relaxation for inference in natural language processing},
year = {2012},
issue_date = {September 2012},
publisher = {AI Access Foundation},
address = {El Segundo, CA, USA},
volume = {45},
number = {1},
issn = {1076-9757},
journal = {J. Artif. Int. Res.},
month = sep,
pages = {305–362},
numpages = {58}
}
